\documentclass[a4paper,fleqn]{cas-dc}

\usepackage[numbers]{natbib}
\usepackage{amsmath,amssymb,amsfonts}
\usepackage{algorithmic}
\usepackage[english]{babel}
\usepackage{algorithm} 
\usepackage{algorithmic}
\usepackage{graphicx}
\usepackage{textcomp}
\let\labelindent\relax
\usepackage[shortlabels]{enumitem}
\usepackage{xcolor}
\usepackage{amsmath}
\usepackage[textsize=footnotesize,textwidth=1.5cm]{todonotes}
\newtheorem{theorem}{Theorem}

\newtheorem{prop}{Proposition}
\newtheorem{lemma}{Lemma}
\newtheorem{definition}{Definition}
\newtheorem{assumption}{\textbf{Assumption}}
\newtheorem{example}{\indent Example}
\newproof{proof}{Proof}
\newproof{pot}{Proof of Theorem \ref{thm}}

\def\and{\textrm{ and}}

\newcommand{\Domain}{\mathcal{D}}

\newcommand{\PWA}{\mathrm{PWA}}

\newcommand{\genpwaparamM}{A}
\newcommand{\genpwaparamV}{a}

\newcommand{\Mode}{I}
\newcommand{\state}{x}

\newcommand{\pState}{X}

\newcommand{\Vertex}{\mathcal{F}_0}

\newcommand{\prtition}{\mathcal{P}}

\newcommand{\qState}{Z}

\newcommand{\vMatrix}{E}
\newcommand{\vVec}{e}

\newcommand{\Rprtition}{\mathcal R}

\newcommand{\Pindex}{\Mode(\prtition)}

\newcommand{\Int}[1]{\mathrm{Int}\left(#1 \right)}
\newcommand{\Dom}[1]{\mathrm{Dom}\left(#1 \right)}
\newcommand{\conv}[1]{\mathrm{conv} \left(#1 \right)}

\newcommand{\R}{\mathbb{R}}

\usepackage{hyperref}

\begin{document}
\let\WriteBookmarks\relax
\def\floatpagepagefraction{1}
\def\textpagefraction{.001}

\shorttitle{Invariant Set Estimation for Piecewise Affine Systems}

\shortauthors{P Samanipour et~al.}

\title [mode = title]{Invariant Set Estimation for Piecewise Affine Dynamical Systems Using Piecewise Affine Barrier Function}                      



%
\author[1]{Pouya Samanipour}[
                        orcid=0000-0002-7724-8878]






\author[1]{Hasan A.Poonawala}[orcid=0000-0003-4201-4745]
\cormark[1]
\affiliation[1]{organization=University of Kentucky,
    city={Lexington},
    state={Kentucky},
    country={USA}}
\ead{hasan.poonawala@uky.edu}






\cortext[cor1]{Corresponding author}



\begin{abstract}
This paper introduces an algorithm for approximating the invariant set of closed-loop controlled dynamical systems identified using ReLU neural networks or piecewise affine ($\PWA$) functions, particularly addressing the challenge of providing safety guarantees for ReLU networks commonly used in safety-critical applications. The invariant set of $\PWA$ dynamical system is estimated using ReLU networks or its equivalent $\PWA$ function. This method entails formulating the barrier function as a $\PWA$ function and converting the search process into a linear optimization problem using vertices. We incorporate a domain refinement strategy to increase flexibility in case the optimization does not find a valid barrier function. Moreover, the objective of optimization is to maximize the invariant set based on the current partition. Our experimental results demonstrate the effectiveness and efficiency of our approach, demonstrating its potential for ensuring the safety of $\PWA$ dynamical systems.
\end{abstract}



\maketitle
\section{Introduction}
Safety is paramount in various applications, such as robotics and autonomous driving. It has become increasingly common for these applications to be deployed in diverse environments as they gain more advanced features and autonomy. It is imperative to guarantee safety for both robots or autonomous vehicles and the humans engaging with them~\cite{liu2023safe}. Conventional control methods often rely on pre-established models and manually designed controllers, a paradigm that might not be optimal in the face of the complex and dynamic environment~\cite{anand2021safe}. 

One common approach in the literature that can guarantee safety automatically is the model predictive controller (MPC)~\cite{morari1999model}. MPC consists of making decisions based on a predictive model of the system~\cite{morari1999model}. Control sequences must be optimized over a finite time horizon while considering the system's dynamics and constraints. In the context of the MPC controller, due to the fact that it heavily depends on the accuracy of the predictive model, learning-based approaches are used in order to learn the predictive model~\cite{manzano2020robust,seel2021neural,MAKDESI2023100849}.

The other common approach for guaranteeing safety automatically is using barrier functions (BF). To verify safety in dynamical systems, a barrier function synthesized by the sum of squares was initially proposed by~\cite{prajna2007framework}. In this case, the BF is dealing with a closed-loop dynamic with a known safe controller. For dealing with dynamical systems with unknown safe controllers, as described in~\cite{wieland2007constructive}, the control barrier function (CBF) was initially introduced, followed by reciprocal and zeroing less conservative barrier functions described in~\cite{ames2014control,ames2019control}. CBFs are commonly employed as safety filters in control systems. 
Since these methods are highly dependent on the model, the barrier function should be robust against uncertainty~\cite{HAMDIPOOR2023100840}.

Accordingly, the use of machine learning approaches to learn the model and the control policy using data has gained significant attention in recent years~\cite {anand2021safe,jebellat2021training,shamsoshoara2023joint,rajoli2023triplet,kargar2023integrated,10354048}. 
Learning methods have led to an increased interest in rectified linear unit (ReLU) neural networks (NN). This is due to the universal approximation theorem of ReLU neural networks~\cite{leshno1993multilayer, Huang2020relu}. ReLU NNs are widely used to learn dynamics and control policy~\cite{dai2021lyapunov,10108069,sun2022formal}.
Learning approaches are challenged by the importance of ensuring safety during training and testing~\cite{marvi2021safe}. Thus, it is necessary to develop an automated approach for determining the forward invariance set for learned dynamics and control policy. 

Supervised learning and reinforcement learning (RL) can be used to ensure safety. A model-free RL controller is combined with a model-based controller using CBF in ~\cite{cheng2019end}. The method utilizes a Gaussian Process (GP) model to learn unknown system dynamics and derives CBFs from the GP model. 
An off-policy actor-critic RL algorithm is employed in~\cite{marvi2021safe} to maintain a safe policy without prior knowledge of the dynamics of the system. The CBF formulation of~\cite{taylor2020learning} accommodates a variety of uncertainties learned episodically by NN by placing uncertainty directly into its formulation. An adaptive sampling method is used to expand the safe region of a dynamical system with an unknown disturbance term in~\cite{wang2018safe}. Through chance-constrained optimization problems,~\cite{khojasteh2020probabilistic} provides a more general approach to optimizing system behavior and ensuring system safety. 

The issue of effectively addressing the synthesis of control barrier functions for broad categories of systems remains a gap in prior research endeavors. Nonetheless, recent investigations have introduced various methodologies~\cite{gillula2012guaranteed,robey2020learning,rabiee2023soft,srinivasan2020synthesis,zhao2021synthesizing,safari2023time} to confront this challenge. The reachability analysis has become a popular method for determining the safe set~\cite{gillula2012guaranteed}, despite potential computational difficulties or a tendency to produce overly conservative estimates. In~\cite{srinivasan2020synthesis}, a support vector machine is leveraged to parameterize the CBF in real-time, while~\cite{taylor2020learning} introduces an optimization-based approach to acquiring knowledge about uncertainty in CBF constraints through expert demonstrations. Another approach, presented in~\cite{rabiee2023soft,10156245}, involves employing a soft minimum and soft max barrier function over a finite horizon to synthesize safe control for nonlinear systems subject to actuator constraints. The utilization of a feedforward neural network with ReLU activation function is proposed in ~\cite{zhao2021synthesizing} for barrier certificates. The verification of the barrier certificate is addressed using Mixed Integer Linear Programming (MILP) and Mixed Integer Quadratically Constrained Problems (MIQCP) with non-linear terms. Furthermore, ~\cite{zhang2023exact} discusses verifying the ReLU neural network in the context of the barrier function. The computational complexity of the verifier, however, remains a challenge.


This paper introduces a novel framework for determining the forward invariant set of $\PWA$ dynamical systems or their equivalent ReLU neural networks, bypassing the need for verification procedures. 
The contributions can be described as follows:
\begin{enumerate}
\item Introduction of a framework for determining the forward invariant set of PWA dynamical systems or ReLU neural networks without relying on traditional verification procedures.
\item The BF is represented as a $\PWA$ function with a structure similar to dynamical systems, with constraints expressed as linear equations.
\item Adding a refinement step to the algorithm to enhance flexibility through the use of higher capacity $\PWA$ Barrier Functions.
\end{enumerate}
The paper's organization is as follows: Section \ref{sec:prelim} provides an overview of relevant concepts, Section \ref{sec:Problem Formulation} details the parameterization of the barrier function and assumptions for the forward invariant set, Section \ref{sec:Main} presents the main algorithm, and Section \ref{sec:res} showcases simulation results.
\section{Preliminaries} \label{sec:prelim}
In this paper, we strive to estimate the invariant set for dynamical systems described by $\PWA$ functions. We first introduce the concept of positive invariant set, barrier function, and then describe the $\PWA$ function and required notation for the remainder of the paper.
\subsection{Forward Invariant Set and Barrier Function Description}
Let us consider the following nonlinear dynamics:
\begin{equation}\label{eq:general nl}
\dot x = f_{cl}(x).
\end{equation}
We assume that the closed-loop dynamics, represented by $f_{cl}$, is locally Lipschitz continuous. Under this assumption, for any given initial condition $x_0 \in \mathbb{R}^n$, there exists a time interval denoted as $I(x_0) = [0, \tau_{\text{max}})$, within which a unique solution $x:I(x_0) \rightarrow \mathbb{R}^n$ exists. This solution satisfies the differential equation \eqref{eq:general nl} with the initial condition $x_0$~\cite{ames2019control}.
\begin{definition}[Forward invariant set~\cite{ames2019control}]
In the case of general controlled nonlinear dynamics described by equation~\eqref{eq:general nl}, let us consider a super level set $C$ defined by a continuous differentiable function $h:D\subset \mathcal{R}^n \rightarrow\mathcal{R}$ as follows:
\begin{equation}\label{eq:BF Domain}
C = \{x \in D : h(x) \geq 0\},
\end{equation}
where $D$ represents the domain of $h$. The set $C$ is said to be forward invariant if, for every $x_0 \in C$, the solution $x(t)$ satisfies $x(t) \in C$ for all $x_0$ and all $t \in I(x_0)$. If the set $C$ is forward invariant, it implies that the system described by equation ~\eqref{eq:general nl} is safe with respect to $C$.
\end{definition}
\begin{definition}[Barrier function~\cite{ames2019control}]
Let $C \subset D \subseteq \mathbb{R}^n$ be the superlevel set of a continuously differentiable function $h(x)$. We define $h(x)$ as a Barrier function if there exists an extended class $K_\infty$ function $\alpha$ such that:
\begin{equation}\label{eq:barrier general}
L_{f_{cl}} h(x) \geq -\alpha(h(x)), \quad \text{for all } x \in D,
\end{equation}
where $L_{f_{cl}} h(x)$ represents the Lie derivative of $h$ along the closed-loop dynamics $f_{cl}$. The equation ~\eqref{eq:barrier general} is the generalized form of Nagumo's theorem~\cite{nagumo1942lage}. Equation ~\eqref{eq:barrier general} implies that $h$ can increase or decrease along the trajectory of the system described by equation ~\eqref{eq:general nl} within $Int(C)$. Additionally, it ensures that $h$ remains increasing along the trajectories of ~\eqref{eq:general nl} on  $\partial C$ as well as on the complement of $C$ within $D$. 
\end{definition}

\subsection{Representations of Piecewise Affine Functions}
\label{sec:pwa}
First, we introduce the notation and definitions for $\PWA$ functions that will be used throughout this paper, and then we describe the standard representation of $\PWA$ functions~\ref {sec:pwafun}. These definitions and notations establish a consistent framework for the subsequent discussions and analyses in the paper.
\paragraph*{\bf Notation}
Given a set $S$, the set $\Mode(S)$ is defined as the index set corresponding to each element in $S$. The convex hull of $S$ is represented by $\conv S$. The interior of $S$ is represented as $\text{Int}(S)$. The boundary of $S$ is expressed as $\partial S$. The closure of $S$ is denoted as $\overline{S}$, encompassing both $\text{Int}(S)$ and its boundary. In the context of vectors, $A^T$ denotes the transpose of $A$. The inner product of two vectors is denoted by $\langle \cdot, \cdot \rangle$. Additionally, it is important to note that the symbols $\preceq$, $\succ$, and $\prec$ carry the same element-wise interpretation as the symbols $\leq$, $>$, and $<$, respectively. 
\begin{definition}
This paper introduces the concept of a partition $\mathcal{P}$, which is defined as a collection of subsets $\{\pState_i \}_{i \in \Pindex}$, where each $\pState_i$ represents a closed subset of $\mathcal{R}^n$ for all $i\in \Pindex$. In partition $\mathcal{P}$, $\text{Dom}(\mathcal{P})=\cup_{i\in \Pindex}\pState_i$  and $\Int{\pState_i} \cap  \Int{\pState_j} = \emptyset$ for $i\neq j$. As a result of these notations and definitions, we can proceed with the introduction of the $\PWA$ functions.
\end{definition}

\subsection{Piecewise Affine Functions}
\label{sec:pwafun}
We represent a $\PWA(x)$ using an explicit parameterization with respect to a partition $\mathcal{P} = \{\pState_i\}_{i \in \Pindex}$, a collection of matrices $\mathbf{A}_\prtition = \{\genpwaparamM_i \}_{i \in \Pindex}$, and vectors  $\mathbf a_\prtition = \{\genpwaparamV_i \}_{i \in \Pindex}$. The parameterization is defined as follows:
\begin{align}
\PWA(x)	&= \genpwaparamM_{i} \state + \genpwaparamV_{i}, \textrm{ if } \state \in \pState_i,  \text{ where}\label{eq:definePWAfun}\\
  \pState_i &= \{x \in \R^n \colon \vMatrix_i \state + \vVec_i  \succeq 0 \} \label{eq:generalpwaparams}.\nonumber
\end{align}
This parameterized formulation concisely represents the piecewise affine function $\PWA(x)$ based on the given partition $\mathcal{P}$ and the associated matrices and vectors.

\begin{assumption}
In this paper, we specifically assume that all $\PWA$ functions presented herein, with their explicit form, satisfy the continuity constraints described in ~\cite{poonawala2021training,johansson1999piecewise}).    
\end{assumption}
\begin{assumption}\label{ass:assumption bounded polytopes}
In our study, all cells are assumed to be bounded. The vertex representation of the $\PWA$ dynamics can therefore be used. In a cell $X_i$, the vertices represent the faces of 0 dimensions. Hence, $X_i$ can be represented mathematically as the convex hull of its vertex set:
\begin{equation}
X_i = \conv{\Vertex(X_i)},
\end{equation}
 where $\Vertex(X_i)$ refers to the set of vertices that belong to the cell $X_i$ ~\cite{henk2017basic}.
\end{assumption}
\begin{assumption}
This paper assumes that the $$\textbf{0}\in X_i \Rightarrow \textbf{0}\in\Vertex(X_i).$$  
\end{assumption}
\section{Problem formulation}\label{sec:Problem Formulation}
The objective of this paper is to identify a forward invariant set for dynamical systems characterized by the equation:
\begin{equation}\label{eq:pwa dynamic}
    \dot x=\PWA(x)
\end{equation}
where $x$ belongs to the $\Domain\subset\mathcal{R}^n$ representing the state variables and the term $\PWA$ denotes a piecewise affine function, which was elaborated in Section \ref{sec:pwafun}. To align with the notation used for the concepts of forward invariance and barrier functions, we define $\Domain$ as:
\begin{equation}\label{eq:Domain}
\Domain \triangleq \{x \in \mathcal{R}^n : x \in \mathrm{Dom}(\mathcal{P})\}.
\end{equation} 
It is important to note that no assumption is made regarding the forward invariance of $\Domain$.
Specifically, this study is concerned with continuous $\PWA$ functions based on polytopes. 
The objective of this section is to identify the invariant set for~\eqref{eq:pwa dynamic} which represents a closed-loop system under a designed controller for which an invariant set exists, but is initially unknown.Various methods can be employed to design the controller, such as traditional control techniques and explicit model predictive control (eMPC).

\subsection{Piecewise affine barrier function:}\label{sec:PWA BF}
This section introduces the formulation of the BF as a $\PWA$ function for the dynamical system described by equation~\eqref{eq:pwa dynamic}, considering a given partition $\mathcal{P}$. 
We consider the following index set.
\begin{align}
I_{\partial \Domain} &= \{ i\in I \colon X_i \cap \partial\Domain\neq \emptyset\}.
\end{align}
The set $I_{\partial \Domain}$ represents the indices of cells containing the boundary of $\Domain$. These cells have points in common with the boundary of the domain $\Domain$. 
In order to determine which vertices are in the $\partial \Domain$ and which vertices are in the $Int(\Domain)$, we can use the sets $I_{b}$ and $I_{int}$ as follows:
\begin{align}
I_{b} &= \{(i,k) : v_k \in \partial\Domain, i \in I_{\partial D}, v_k \in \Vertex(X_i)\}\\
I_{int} &= \{(i,k):v_k \notin \partial\Domain, i \in I(\prtition), v_k \in \Vertex(X_i)\}.
\end{align}
The set $I_{b}$ and $I_{int}$ consists of pairs of indices, where the first element denotes the cell index $i$, and the second element represents the vertex index $k$. $I_b$ represents the indices of the vertices on the boundary and their associated cells, while $I_{int}$ represents the indices of vertices in the interior of the partition and their associated cells. In Figure~\ref{fig:set description}, samples for these index cells are provided.
\begin{figure}
    \centering
    \includegraphics[scale=0.5]{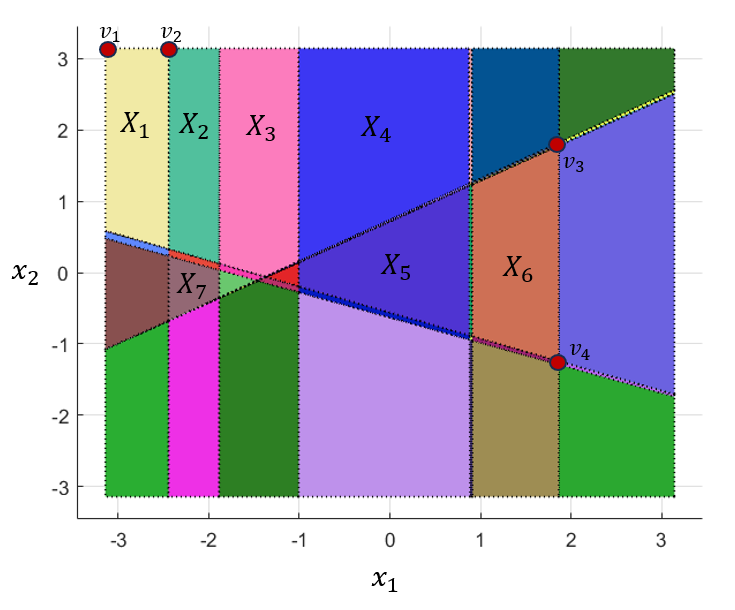}
    \caption{A sample partition of a $\PWA$ dynamics. Cells $X_1$ ,$X_2$, $X_3$ and $X_4$ share points with the boundary, hence, $\{1,2,3,4\}\in I_{\partial \Domain}$. While cells $X_5$,$X_6$, and $X_7$ do not share any points with the boundary. As a result $\{5,6,7\}\notin I_{\partial \Domain}$. Furthermore, $v_1,v_2$ are vertices on the boundary, while $v_3$ and $v_4$ are not on the boundary. Hence, $\{(1,1),(1,2),(2,2)\}\in I_b$ and $\{(6,3),(6,4)\}\in I_{int}$.}
    \label{fig:set description}
\end{figure}

The main idea in this paper is to formulate the barrier function using a $\PWA$ formulation for the cell $X_i$ as follows:
\begin{equation}\label{eq: BF PWA}
h_i(x)=p_i^Tx+q_i \quad \text{for} \quad i \in I(\prtition) 
\end{equation}  
where $p_i\in \mathcal{R}^n$ and $q_i$ is a scalar.
The function $h(x)$ is continuous and differentiable within the interior of the cell. This allows us to compute the derivative of the candidate Barrier function along the dynamics described by ~\eqref{eq:general nl} within the interior of the cell $X - \{0\}$. The derivative is expressed as:
\begin{equation}
\mathcal{L}_f h = \langle \nabla h, f(x) \rangle, \label{eq:generallyapcond}
\end{equation}
where $\nabla h$ represents the gradient of $h(x)$, and $\mathcal{L}$ represents the Lie derivative.
\begin{assumption}
    Let's assume the candidate barrier function $h_i(x) = p_i^T x + q_i$ and the dynamics equation given by $\dot{x} = A_ix + a_i$ for all $x \in \pState_i$. Then the derivative of the barrier function along the trajectories at $\mathbf{x}'$, when $h(\mathbf{x}')$ is differentiable can be computed as:
\begin{equation}\label{eq:derivative definition}
\dot{h} = p_i^T(A\mathbf{x}' + a).
\end{equation}
\end{assumption}

This paper has the following objectives:
\begin{itemize}
    \item develop an optimization problem to find the invariant set for the~\eqref{eq:pwa dynamic}.\label{goal1}
    \item Maximize the forward invariant set obtained from $\PWA$ barrier function.\label{goal2}
\end{itemize}
To achieve these goals, we must form the constraints based on the requirements of the barrier function as described in~\eqref{eq:barrier general}. In order to meet the goal of finding a maximized invariant set using a barrier function formulated as a $\PWA$ function, the following propositions and theorems are presented.
\begin{prop}\label{prop:neg on bound}
The following are equivalent:
\begin{enumerate}[(a)]
    \item $C \subseteq \Domain$
    \item  $p_i^Tv_k + q_i \leq 0,\quad \forall(i,k)\in I_b$\label{prop:neg on bound 2}
\end{enumerate} 
\end{prop}
\begin{proof}
Regarding the definition $C$ in ~\eqref{eq:BF Domain} and (a), if $C=D$ we can conclude that  $$p_i^Tv_k + q_i=0,\quad (i,k)\in I_b.$$ Furthermore, if $C\subset D$, we may find points $x_d\in D$ where $x_d\notin C$. In light of Assumption~\ref{ass:assumption bounded polytopes} and parametrization of the barrier function, ~\eqref{eq: BF PWA}, it can be concluded that there is a vertex $v_{d}$ at which $p_i^Tv_d+q_i<0$. Accordingly, we can conclude (b).   
\end{proof}

It is essential that the obtained invariant set must be a subset of $\Domain$. Using Proposition~\ref{prop:neg on bound}, we can satisfy this requirement by maintaining the barrier function non-positive on the $\partial \Domain$ of the $\Domain$. 
However, there is no guarantee that such $p_i$ and $q_i$ values exist. In order to ensure the feasibility of the optimization problem, \ref{prop:neg on bound 2} must be relaxed to:
\begin{equation}\label{eq:soft negative}
p_i^Tv_k + q_i-\tau_{b_i} \leq 0,\hspace{0.2em} (i,k)\in I_b,v_k\in\Vertex(X_i),
\end{equation}
where $\tau_{b_i}$ represents a positive slack variable introduced to ease the constraint. This slack variable allows us to identify cells $X_i$ that might necessitate a higher capacity $\PWA$ function. The process for increasing the capacity of the $\PWA$ function is explained in section \ref{sec:Ref}.    
\begin{prop}\label{prop:continuity}
The candidate $\PWA$ barrier function is continuous if and only if the following condition is met:
\begin{align}\label{eq:cont}
h_i(v_k) = &h_j(v_k),\quad i \neq j\in \Pindex,\nonumber\\
&v_k\in \Vertex(X_i)\cap \Vertex(X_j).
\end{align}
\end{prop}
\begin{proof}
    Clearly, if the barrier function is continuous, then the barrier function value must be the same at all common points, such as vertices.     
    Conversely, if~\eqref{eq:cont} is satisfied due to the parametrization of the barrier function and convexity of the common faces between any two adjacent cells, we can conclude that the candidate barrier function is continuous on the boundary of the cells. As stated in~\eqref{eq: BF PWA}, the candidate barrier function is continuous throughout the cell's interior. Therefore, the candidate $\PWA$ barrier function is continuous.
\end{proof}
\begin{prop}\label{prop: Generalized Nagumo}
The following are equivalent if we assume $\alpha(x)=\Tilde{\alpha} x$, where $\Tilde{\alpha}>0$:
\begin{align}\label{eq: Generalized Nagumo}
    (a)&\quad\frac{\partial h}{\partial x}\dot{x}+\alpha(h(x))\geq0\quad \forall x \in \Domain \nonumber\\
    (b)&\quad p_i^T(A_iv_k+a_i)+\Tilde{\alpha}(p_i^Tv_k+q_j)\geq0, \nonumber
    \forall i \in \Pindex,\\ &\quad k\in \Vertex(X_i).
\end{align}
\end{prop}
\begin{proof}
    Since the cells are convex polytopes, we need only check (a) at all vertices to ensure that all points within the interior of the cells satisfy (a). Therefore, (b) can be derived directly from (a) by using~\eqref{eq:pwa dynamic},~\eqref{eq:derivative definition} and assumption $\alpha(x)=\Tilde{\alpha} x$ for all vertices. %
    These conditions also account for the non-differentiable nature of $h$ at boundaries of cells, however we suppress the fairly involved details, resulting in an abuse of notation. %
    Our prior work in~\cite{10108069,10313502} discusses how the theory of set-valued Lie derivatives and continuous PWA functions justifies this claim. %
    If the PWA dynamics involved switching, this claim would no longer hold. 
\end{proof}
\begin{theorem}
    A barrier function~\eqref{eq: BF PWA} is considered to be a valid barrier function if Proposition~\ref{prop:neg on bound}, Proposition~\ref{prop:continuity} and Proposition~\ref{prop: Generalized Nagumo} hold.  
\end{theorem}
\begin{proof}
    It is guaranteed by the first proposition that the obtained invariant set is a subset of the partition domain $\Domain$. The barrier function will be continuous as we used the proposition~\ref{prop:continuity}, and it will be differentiable within the cell's interior. Furthermore, Proposition~\ref{prop: Generalized Nagumo} guarantees the generalized Nagumo's condition. 
\end{proof}

Using the propositions presented so far, the first objective has been achieved. 
In order to achieve the second goal, maximizing the positive invariant set with respect to the current partition, we need to ensure that as many vertices in the interior of the domain as possible are also inside the invariant set or, in other words, $h_i(v_k)\geq0$ where $(i,k)\in I_{int}$. 
\begin{assumption}
    To achieve the second objective, the following inequality can be utilized:
\begin{equation}\label{eq:relaxed positive interior}
p_i^Tv_k + q_i + \tau_{int_i} > 0,\hspace{0.2em} (i,k) \in I_{int}.
\end{equation}
\end{assumption}
In~\eqref{eq:relaxed positive interior}, $\tau_{int_i}$ represents a slack variable introduced to relax the constraint. By incorporating this slack variable, some flexibility is provided. The importance of this flexibility arises from the fact that not all interior cells may be considered part of the invariant set. Including the slack variable makes it possible to identify interior cells that need to be refined to achieve the desired forward invariant set. In~\ref{sec:opt}, it will be shown how this constraint will help us to maximize the invariant set.     




By integrating these three Propositions and~\eqref{eq:relaxed positive interior}, we are able to construct a systematic approach to searching for the invariant set using the barrier function. 
\section{Main Algorithm}\label{sec:Main}
According to the assumptions and propositions made in the previous section, the main algorithm will be constructed. As a first step, the optimization problem is formulated, and then the refinement process is described to increase the capacity of the $\PWA$ function. In the end, we will discuss the search process. 
\subsection{Forming optimization problem}\label{sec:opt}
We can create an optimization problem to derive the barrier function using equations~\eqref{eq:soft negative} through~\eqref{eq:relaxed positive interior}. The cost function can be defined as follows:
\begin{equation}    \mathcal{J}=\sum_{i=1}^{M}\tau_{b_{i}}+\sum_{i=1}^{N}\tau_{int_{i}}
\end{equation}
$M$ and $N$ are the number of elements in $I_b$ and $I_{int}$ respectively.

The following represents the ultimate form of this optimization problem:
\begin{align}\label{eq:barrierrelaxfinal}
&\min_{ p_i, q_i,\tau_{int_{i}},\tau_{b_{i}}}  \quad  \mathcal{J}  \\
&\text{Subject to:} \nonumber\\
&p_i^T v_k+q_i-\tau_{b_i}\leq -\epsilon_1, \hspace{0.2em}\forall (i,k) \in I_b\nonumber\\
&p_i^T v_k+q_i+\tau_{int_i}\geq\epsilon_2, \hspace{0.2em} \forall (i,k) \in I_{int}\nonumber\\
&p_i^T(A_iv_k+a_i)+\alpha(p_i^Tv_k+q_i)\geq \epsilon_3, \hspace{0.2em}\forall i \in \Pindex\nonumber\\
&h_i(v_k)=h_j(v_k),  \forall v_k\in \Vertex(X_i)\cap \Vertex(X_j)\nonumber\\
&\tau_{b_i}\geq 0,  \forall i\in I_{\partial \Domain}\nonumber
\end{align} 
where $\epsilon_1,\epsilon_2,\epsilon_3>0$. 
\begin{lemma}
    Optimization problem ~\eqref{eq:barrierrelaxfinal} is always feasible.
\end{lemma}
\begin{proof}
As a result of the construction of the optimization problem, this result is obtained. 
\end{proof}

Although the optimization problem is always feasible, The solution will be the valid control barrier function if and only if:
\begin{equation}
\sum_{i=1}^{M}\tau_{b_{i}}=0.
\end{equation}
If a non-zero $\tau_{b_i}$ exists, it indicates that the obtained solution is not a valid barrier function. In other words, based on the current solution, no forward invariant set $C$ satisfies the condition $C \subseteq \Domain$. In order to address this issue, a refinement process is proposed. Refinement involves dividing each cell into smaller subcells, allowing for the construction of higher-capacity $\PWA$ barrier functions. The details of the refinement procedure will be discussed in the following sections.
\begin{lemma}\label{lemma: maximal}
    The valid invariant set obtained from the optimization problem~\eqref{eq:barrierrelaxfinal} is the largest possible invariant set with the current partition and $\PWA$ barrier function.
\end{lemma}
\begin{proof}
    Let's consider that there exists an invariant set $C'$ where $C\subset C'$. 
    As a result, there are points in a number of cells that belong to $C'$ but are not included in $C$. For simplicity, let's assume the difference of invariant set $C'$ and $C$ is only at cell $X_1$ at some points such as $x_1$. This assumption does not pose any problems for generalization.At $x_1$, the barrier function corresponding to $C'$, $h'(x)$, is positive, whereas the barrier function corresponding to $C$, $h(x)$, is negative. Therefore, we can conclude that $h_1'(x_1)>h_1(x_1)$. Due to the parametrization of the barrier function as a $\PWA$ function, and $h_1'(x_1)>h_1(x_1)$, there exists a $v_k\in \Vertex(X_i)$ where $h_1'(v_k)>h_1(v_k)$. Regarding~\eqref{eq:relaxed positive interior}, there exist a $\tau'_{int_1}<\tau_{int_1}$, where $\tau'_{int_1}$ and $\tau_{int_1}$ are the slack variables in~\eqref{eq:relaxed positive interior} for $h_1'(x)$ and $h_1(x)$ respectively. Moreover, according to the assumption $C\subset C'$ we can conclude that the slack variable at all the other interior cells is $$\tau'_{int_i}\leq \tau_{int_i} \quad \forall i \in I_{int}-\{1\}.$$ Therfore, $\sum_{i=1}^{N}\tau_{int_{i}}>\sum_{i=1}^{N}\tau'_{int_{i}}$ which is contrary to the minimization of the slack variable $\tau_{int_i}$ in the optimization problem~\eqref{eq:barrierrelaxfinal}. As a result, there is no invariant set associated with the current partition where $C\subset C'$.  
\end{proof}
\begin{algorithm}[tb]
    \begin{algorithmic}
    \REQUIRE  $\PWA(x)$, Vertices
    \STATE{Solve Optimization Problem ~\eqref{eq:barrierrelaxfinal} with the initial $\PWA$ dynamics.}
    \WHILE{$\Sigma_{i=1}^{N}\tau_{b_i}\neq 0$ and $\textit{computational-time}<3600s$}  
    \FOR{$\tau_{b_i}\neq0$ and $\tau_{int_i}\neq0$} \STATE{1- Refine $X_i$ using vector field refinement Section~\ref{sec:Ref}.}\ENDFOR 
    \STATE{Solve Optimization Problem \eqref{eq:barrierrelaxfinal} with the refined $\PWA$ dynamics}
	\ENDWHILE
    \RETURN $h(x) = p_i^Tx+q_i    \quad \forall i\in I(\prtition) $.
    \end{algorithmic}
    \vspace{1em}
    \caption{Barrier function search using vertices with known control}
    \label{alg:refinesearch}    
\end{algorithm}

\subsection{Refinement}\label{sec:Ref}
A refinement of the current partition aims to increase the flexibility of the Barrier function search process. In mathematical terms, given two partitions $\prtition = \{Y_i\}_{i \in I}$ and $\Rprtition = \{ \qState_j\}_{j \in J}$ of a set $S = \Dom{\prtition} = \Dom{\Rprtition}$, we say that $\Rprtition$ is a refinement of $\prtition$ if $\qState_j \cap Y_i \neq \emptyset$ implies that $\qState_j \subseteq Y_i$.
To increase flexibility, we introduce a refinement on the cell $X_i$ with a nonzero slack variable, $\tau_{b_i}$. As a result of the refinement of the cell $X_i$, consider that two new sub-cells will be created, $X_{i_1}$ and $X_{i_2}$. Each sub-cell can be parameterized with a $\PWA$ Barrier function, denoted as $h_{i_1}=p_{i_1}^Tx+q_{i_1}$ and $h_{i_2}=p_{i_2}^Tx+q_{i_2}$, respectively. By dividing the cell into smaller sub-cells and introducing separate parameterizations, the candidate Barrier function for cell $X_i$ achieves a higher capacity $\PWA$ function, thereby increasing flexibility in the Barrier function search process.
This paper's refinement technique includes two steps: choosing a new vertex and using Delaunay triangulation to form subcells. The first step is choosing a new vertex on the boundary of the considered cell with a nonzero slack variable using the vector field approach~\cite{10313502}. This method is designed to decrease the variation of the vector field of the new subcells. Therefore, the new vertex is chosen in a way where its vector field will be the angle bisector of the edge, the faces of 1 dimensions~\cite{henk2017basic}, with the biggest variation of the vector field. Based on the newly appointed vertex, subcells will be formed using Delaunay triangulation. A Delaunay triangulation will also be helpful in maintaining continuity. For more detail, please see~\cite{10313502}.  
\subsection{Search Algorithm}
The pseudo-code for the search algorithm is presented in \textbf{Algorithm}~\ref{alg:refinesearch}. The objective is to identify a valid Barrier function on the initial partition by performing an optimization, as given in~\eqref{eq:barrierrelaxfinal}. In cases where a valid solution is not obtained initially, the algorithm proceeds to refine the partition defining the barrier function. This refinement is directed towards cells exhibiting $\tau_{b_i}\neq 0$ or $\tau_{int_i}\neq 0$. Subsequently, the selected cells are subdivided into smaller sub-cells, and the optimization process is repeated on the refined partition. The iterative refinement process continues until a valid solution is found or the maximum computation time is reached. 

As a result of the systematic approach outlined by the algorithm, the partition can be refined in a methodical manner while actively searching for a valid Barrier function. Using this iterative approach increases the capacity of the Barrier function, although it may significantly increase the size the optimization problem.

In this search algorithm, the next important issue is how to select $\Tilde{\alpha}$ in~\eqref{eq: Generalized Nagumo}. The literature does not provide a systematic method for selecting the $\kappa_\infty$ function in~\eqref{eq: Generalized Nagumo}. Choosing $\Tilde{\alpha}=0$ will force all level sets of the barrier function to be invariant, resulting in a conservative result. This paper uses the bisection method to find the least conservative $\Tilde{\alpha}$. Bisection is performed on a predefined interval $[\Tilde{\alpha}_1,\Tilde{\alpha}_2]$. The \textbf{Algorithm}~\ref{alg:refinesearch} will be tested by the $\Tilde{\alpha}_1$ and $\Tilde{\alpha}_2$. 
Upon bisection, the greatest $\Tilde{\alpha}$ that satisfies the search algorithm for finding the valid barrier function is selected.
This is done by selecting $\Tilde{\alpha}_2$ so that it violates the algorithm's computation time limit and selecting $\Tilde{\alpha}_1=10^{-3}$ so that it is conservative but still within the algorithm's operational time limit.    
Please note that the search interval will depend on the nature of the problem, and we may need to extend the search interval in order to find the least conservative invariant set. 
\section{Results and simulations}\label{sec:res}
We provide three examples to illustrate the performance of the Barrier function search algorithm presented in \textbf{Algorithm} \ref{alg:refinesearch}. All computations are implemented using Python 3.9 and the Mosek optimization package ~\cite{aps2022mosek} on a computer with a 2.1 GHz processor and 8 GB RAM. A tolerance of $10^{-8}$ is used to determine if a number is nonzero. 
Moreover, the examples employ a tolerance of ${\epsilon_1=\epsilon_2=\epsilon_3=10^{-4}}$.
These examples demonstrate the proposed algorithm's effectiveness and efficiency in finding valid $\PWA$ Barrier functions or the forward invariant set, highlighting the algorithm's performance in different scenarios.
\begin{example}[Inverted Pendulum~\cite{rabiee2023soft}]
The inverted pendulum system can be described by the following equations ~\cite{rabiee2023soft}:
\begin{align}
\dot{x}_1 &= x_2 \nonumber\\
\dot{x}_2 &= \sin(x_1) + u\nonumber
\end{align}
where $x_1$ represents the angle of the pendulum and $x_2$ is its angular velocity. The control input $u$ is given by $u = 3x_1 + 3x_2$ and is subject to saturation between the lower bound $-1.5$ and the upper bound $1.5$.
The domain $\Domain$ of the system is defined as $\Domain=\pi - ||x||_\infty\geq 0$, where $||x||_\infty$ represents the infinity norm of the state vector $x$. 
We approximate the right hand side of the dynamics using a ReLU neural network which is equivalent to $PWA(x)$. 
To train a neural network for this system, a dataset of $10000$ samples is used. The neural network is a single-hidden layer ReLU network with eight neurons. The training is performed using PyTorch 2.0.0 with ADAM optimization.

For this problem, we have set $\Tilde{\alpha}=0.2$. \textbf{Algorithm} \ref{alg:refinesearch} started with 34 cells and ended up with 680 due to the refinement. The forward invariant set is approximated, as shown in Fig.\ref{fig:Barrier function}. Moreover, A sample trajectory with the initial conditions $x_0=[-1,3]^T$ is shown in Figure \ref{fig:Sample trajectory_IP}. Additionally, the value of the barrier function for the specified trajectory is shown in Fig.\ref{fig:Barrier value}. 
Details about this example can be found in Table.\ref{table:summary}. The computation time is based on the average time it takes to find the solution after it has been run ten times in Table.\ref{table:summary}.
\begin{figure}
    \centering
    \includegraphics[scale=0.5]{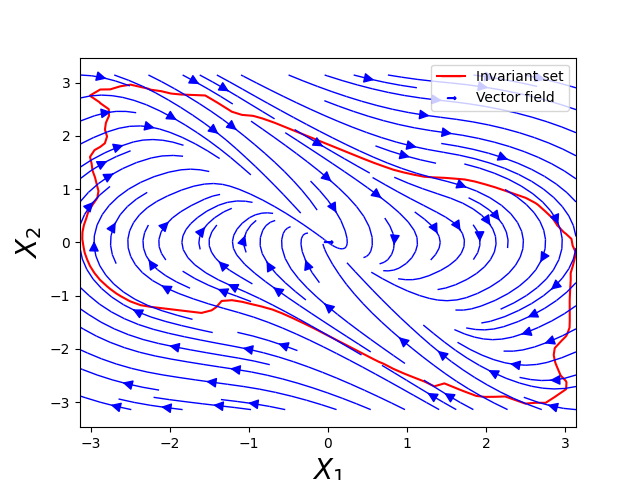}
    \caption{Vector fields for the Inverted Pendulum identified using ReLU NN and its estimated Invariant set}
    \label{fig:Barrier function}
\end{figure}
\begin{figure}
    \centering
    \includegraphics[scale=0.5]{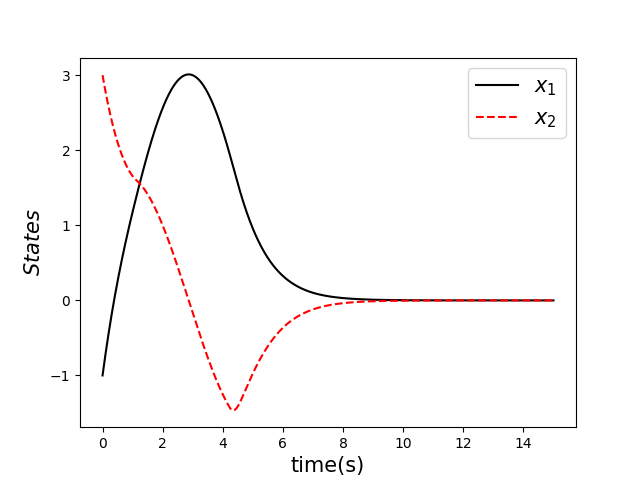}
    \caption{$x_1$ and $x_2$ for a sample trajectory with initial condition $x_0=[-1,3]^T$}
    \label{fig:Sample trajectory_IP}
\end{figure}
\begin{figure}
    \centering
    \includegraphics[scale=0.47]{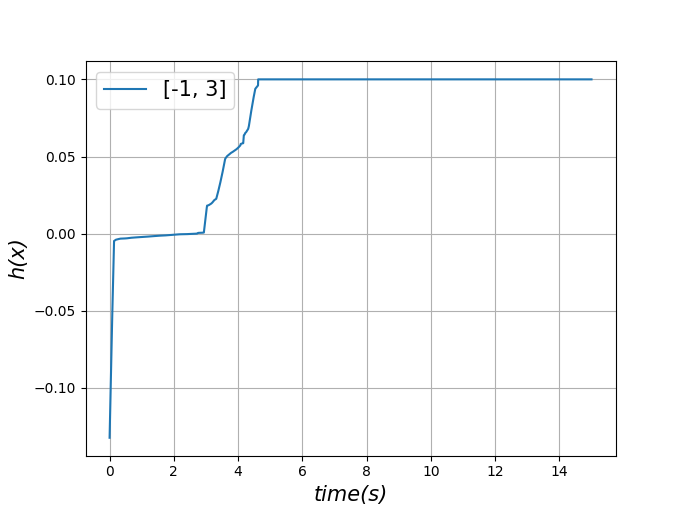}
    \caption{$h(x)$ along the sample trajecty. Since the trajectory starts from outside the invariant set, the value of $h(x)$ is negative. As soon as they enter the invariant set, the value of $h(x)$ becomes positive.}
    \label{fig:Barrier value}
\end{figure}   
\end{example}

\begin{example}[Explicit model-predictive controller~\cite{10313502}]
To formulate a control policy for a discrete-time linear dynamical system, we employ Model Predictive Control (MPC) following the methodology outlined in ~\cite{10313502}. The dynamics of the considered system are governed by the following equations:

\begin{align}
x_{t+1} = \begin{bmatrix}
1 & 1\\
0 & 1
\end{bmatrix}x_{t}+\begin{bmatrix}
1 \\
0.5
\end{bmatrix}u_{t}\label{eq:mpcdiscrete}
\end{align}

To address the MPC problem, we adhere to the specifications presented in ~\cite{10313502}, encompassing the stage cost, actuator constraints, and state constraints. The explicit MPC controller is implemented using the MPT3 toolbox ~\cite{kvasnica2004multi} in the Matlab environment. A discretization step of 0.01 seconds is employed to convert the discrete-time to the continuous-time dynamics.

The effectiveness of the proposed algorithm for the continuous form of ~\eqref{eq:mpcdiscrete} with the explicit MPC controller is validated. The trail and error method is used in this example in order to obtain $\Tilde{\alpha}=4.5$ as the least conservative value. In Fig.\ref{fig:barrier function MPC}, the invariant set of the 2-D MPC is depicted, providing a visual representation of the system's behavior. Additionally, Fig.\ref{fig:barrier value MPC} illustrates the value of the barrier function along sample trajectories for the 2-D MPC. Notably, the algorithm initiates with 26 cells and undergoes four iterations of refinement, resulting in an increased number of cells (136).
\begin{figure}
    \centering
    \includegraphics[scale=0.5]{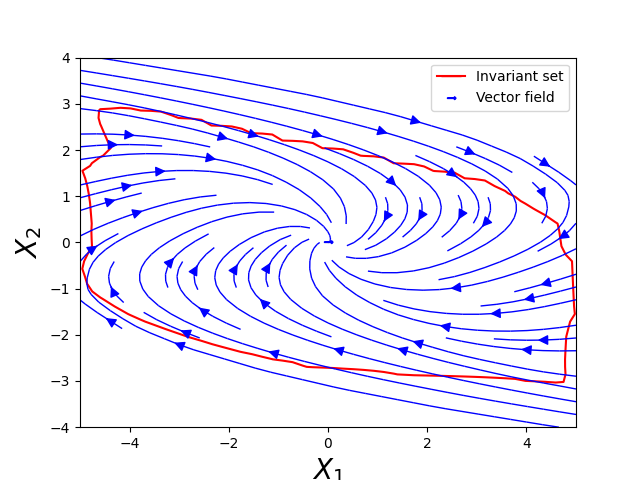}
    \caption{2-D mpc vector fields, invariant set}
    \label{fig:barrier function MPC}
\end{figure}
\begin{figure}
    \centering
    \includegraphics[scale=0.5]{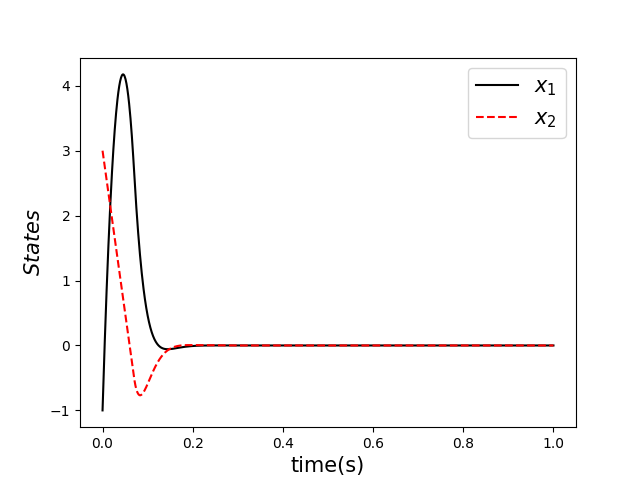}
    \caption{The states for the sample trajectory with the initial condition $x_0=[4,-3.5]^T$.}
    \label{fig:enter-label}
\end{figure}
\begin{figure}
    \centering
    \includegraphics[scale=0.5]{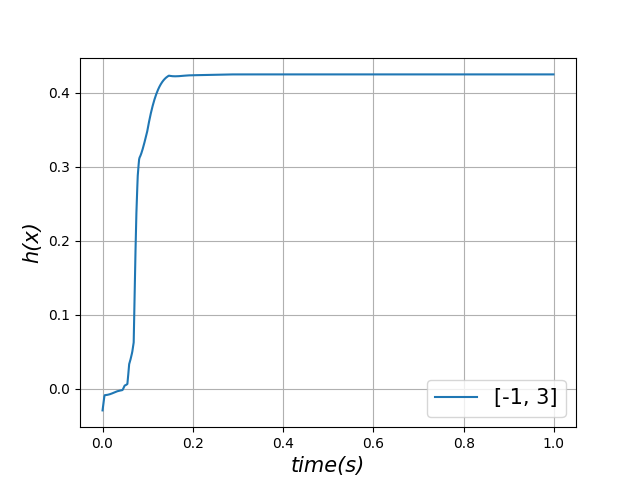}
    \caption{$h(x)$ along the sample trajectory for the 2-D MPC.}
    \label{fig:barrier value MPC}
\end{figure}
\end{example}
\begin{example}[4-D Example ~\cite{chen2021learning}]
\label{example:4D}
In this example, we explore a 4-D Model Predictive Control (MPC) scenario based on the formulation presented in ~\cite{chen2021learning}. The system dynamics are given by:

\begin{align}\label{eq:MPC4D}
    x_{t+1}=&\begin{bmatrix}
      0.4346&-0.2313&-0.6404&0.3405\\
      -0.6731&0.1045&-0.0613&0.3400\\
      -0.0568&0.7065&-0.086&0.0159\\
      0.3511 &0.1404&0.2980&1.0416
  \end{bmatrix}x_t+\\
  &\begin{bmatrix}
      0.4346,-0.6731,-0.0568,0.3511
  \end{bmatrix}^Tu_t.\nonumber
\end{align}

This example inherits the key parameters from ~\cite{chen2021learning}, such as a state constraint of $\lVert x \rVert_{\infty} \leq 4$, an input constraint of $\lVert u \rVert_{\infty} \leq 1$, a prediction horizon of $T=10$, a stage cost of $Q=10I$, and $R=1$. The design of the MPC controller incorporates a state constraint $\lVert x \rVert_{\infty} \leq 4$, an input constraint $\lVert u \rVert_{\infty} \leq 1$. The terminal cost utilizes $P_\infty$, a solution to the Riccati equation defined by $(A,B,Q,R)$, and the terminal set is selected to maximize the positive invariant set.

We find the invariant set of a continuous-time linear system corresponding to~\eqref{eq:MPC4D} when using an explicit MPC controller for feedback. %
We assume that~\eqref{eq:MPC4D} corresponds to a sampling time of $0.01$s. %
The closed loop is a $\PWA$ dynamical system with 193 cells. %
We employ vector field refinement techniques to search for the continuous $\PWA$ Barrier function using \textbf{Algorithm} \ref{alg:refinesearch}. In this example, we found $\Tilde{\alpha}=9.7$ by trial and error. With the vector field refinement, the algorithm successfully converges, finding the solution in 2886 seconds and generating 17689 cells.
\end{example}
\def\mtw{\textwidth}
\begin{table}
\vspace{1em}
\centering
\begin{tabular}{p{0.12\mtw}|p{0.1\mtw}|p{0.08\mtw}|p{0.1\mtw}}
\cline{1-4}
Example & Initial number of cells & final number of cells &  Computational time(s)\\ \hline
Inverted Pendulum  & 34 & 680 & 10.3 \\ \hline
2-D MPC & 26 & 136 & 1.21 \\  \hline
4-D MPC &192& 17689&2886\\ \hline
\end{tabular}	
\caption{Summary of results for finding the valid barrier function using the proposed methods \textbf{Algorithm}.~\ref{alg:refinesearch}.}
\label{table:summary}
\end{table}

\paragraph*{Limitations}
In order to find the $\PWA$ barrier function for the $\PWA$ dynamical system with $n$ state variables, we require  $(n+1)n_r$ optimization variables, where $n_r$ represents the number of cells. In this paper, we added the $M+N$ slack variable as described in \ref{sec:opt}. It has been demonstrated in~\cite{10313502} that $n$ influences the number of cells $n_r$. As described in ~\cite{el2012interior}, the computational complexity of a linear optimization problem is $O(n_p^{3/4}log(\frac{n_p}{\epsilon}))$ where $n_p$ is the number of variables and $\epsilon$ is the accuracy. As a result, $n$ as the dimension of the state space directly impacts the optimization process's computational time. When applying our algorithm to various systems, it is essential to consider these complexities and challenges.  
\section{Conclusion}
This paper presents a novel computational framework to estimate the invariant set for $\PWA$ dynamical systems or their equivalent ReLU neural networks, addressing the challenge of ensuring safety guarantees. 
The key contribution is formulating the barrier function as a $\PWA$ function and casting the search process as a linear optimization problem. The method intelligently refines the domain when the optimization fails, enhancing the search process's flexibility. The optimization problem aims to maximize the invariant set, improving safety guarantees. The proposed method effectively estimates the invariant set through non-trivial examples while maintaining computational efficiency. The results demonstrate the approach's capability and applicability in providing safety guarantees for $\PWA$ dynamics and ReLU neural networks in safety-critical systems. Further exploration and extension of this approach hold potential for more sophisticated techniques in addressing safety concerns in neural network-based systems.
\bibliographystyle{cas-model2-names}

\bibliography{acc2023}
\end{document}